\documentclass[conference]{IEEEtran}
\usepackage[explicit]{titlesec}
\usepackage{graphicx}
\usepackage{easyReview}
\usepackage{verbatim}
\usepackage{epstopdf}
\usepackage{setspace}
\usepackage{mathtools} 
\usepackage{mdframed} 
\usepackage{subfigure}
\usepackage{booktabs}
\usepackage[most]{tcolorbox}
\usepackage{algorithmic}
\usepackage{amsmath}
\usepackage{amsthm}
\usepackage{mathrsfs}
\usepackage{extarrows}
\usepackage{cite}
\newtheorem{theorem}{Theorem}

\usepackage{bbm}
\usepackage[colorlinks=true, citecolor=blue]{hyperref}
\usepackage{cleveref}
\usepackage{amssymb}
\newtheorem{remark}{Remark}

\newcommand{\RNum}[1]{\uppercase\expandafter{\romannumeral #1\relax}}
\newtheorem{lemma}{Lemma}

\newtheorem{proposition}{Proposition}
\newtheorem{problem}{Problem}
\newtheorem{definition}{Definition}

\usepackage{stfloats}
\usepackage{color, soul}
\usepackage{booktabs}
\usepackage{tikz,xcolor}
\usetikzlibrary{arrows.meta,automata, positioning, backgrounds, matrix}
\newcommand{\E}{\mathbb{E}}

\newenvironment{proofsketch}{\begin{proof}[\textit{Proof Sketch}]}{\end{proof}}
\newtcbox{\hiRatio}{on line,
  colback=white, colframe=blue!70!black,
  boxrule=0.6pt, arc=1.5pt,
  left=1pt,right=1pt,top=0.5pt,bottom=0.5pt}
\definecolor{lime}{HTML}{A6CE39}
\usepackage[left=0.68in,right=0.68in,top=0.67in,bottom=0.99in]{geometry}

\titlespacing{\section}{0pt}{1.2ex plus .0ex minus .0ex}{.3ex plus .0ex}
\titlespacing{\subsection}{0pt}{1.2ex plus .0ex minus .0ex}{.3ex plus .0ex}

\makeatletter
\DeclareRobustCommand{\orcidicon}{%
	\begin{tikzpicture}
		\draw[lime, fill=lime] (0,0) 
		circle [radius=0.16] 
		node[white] {{\fontfamily{qag}\selectfont \tiny ID}};    \draw[white, fill=white] (-0.0625,0.095) 
		circle [radius=0.007];    \end{tikzpicture}
	\hspace{-2mm}}
\foreach \x in {A, ..., Z}{%
	\expandafter\xdef\csname orcid\x\endcsname{\noexpand\href{https://orcid.org/\csname orcidauthor\x\endcsname}{\noexpand\orcidicon}}
}

\newcommand*\bigcdot{\mathpalette\bigcdot@{.5}}
\newcommand*\bigcdot@[2]{\mathbin{\vcenter{\hbox{\scalebox{#2}{$\m@th#1\bullet$}}}}}
\makeatother

\usepackage[linesnumbered,ruled,vlined]{algorithm2e}
	
	\IEEEoverridecommandlockouts
	\makeatletter
	\def\@thmnote#1{\textit{[#1]}} 
	\makeatother

\begin{document}
		\title{Taming the Heavy Tail: Age-Optimal Preemption}
		\author{
			Aimin Li, Yiğit İnce, and
			Elif Uysal, \emph{Fellow, IEEE}\\
			\textit{{\color{black}Communication Networks Research Group (CNG)}, EE Dept, METU, Ankara, Turkiye} \\
			\textit{E-mail: aimin@metu.edu.tr; 
            yigit.ince@metu.edu.tr;
            uelif@metu.edu.tr}\vspace{-1em}
			\thanks{Detailed proofs and additional results can be found in \cite{aiminpreemption}. {Code is available at: \textit{https://github.com/AiminLi-Hi/aoi-preemption-heavy-tail}}. This work was supported by the European Union (through ERC Advanced
Grant 101122990-GO SPACE-ERC-2023-A). Views and opinions expressed
are those of the authors only and do not necessarily reflect
those of the funding agencies.}	
			}

		\maketitle
		\allowdisplaybreaks

		 \begin{abstract}
This paper studies a continuous-time joint sampling-and-preemption problem, incorporating sampling and preemption penalties under general service-time distributions. We formulate the system as an impulse-controlled piecewise-deterministic Markov process (PDMP) and derive coupled \emph{integral} average-cost optimality equations via the dynamic programming principle, thereby avoiding the smoothness assumptions typically required for an average-cost Hamilton–Jacobi–Bellman quasi-variational inequality (HJB--QVI) characterization. A key invariance in the busy phase collapses the dynamics onto a one-dimensional busy-start boundary, reducing preemption control to an optimal stopping problem. Building on this structure, we develop an efficient policy iteration algorithm with heavy-tail acceleration, employing a hybrid (uniform/log-spaced) action grid and a far-field linear closure. Simulations under Pareto and log-normal service times demonstrate substantial improvements over AoI-optimal non-preemptive sampling and zero-wait baselines, achieving up to a $30\times$ reduction in average cost in heavy-tailed regimes. Finally, simulations uncover a counterintuitive insight: \textit{under preemption, delay variance, despite typically being a liability, can become a strategic advantage for information freshness.} 
\end{abstract}
		
		\IEEEpeerreviewmaketitle
		
		\section{Introduction}\label{sectionI}
        \subsection{Motivation}
		Age of Information (AoI) has emerged as a fundamental performance metric for real-time sensing, communication, and control systems \cite{yates2021age,kosta2017age}, shifting the focus from how much data is delivered (throughput-oriented) to how fresh and effective the delivered information is (goal-oriented) \cite{uysal2022semantic,gunduz2023timely,li-goal}. In practice, network systems often operate under highly variable and heavy-tailed delays arising from buffering, routing, or retransmissions. To mitigate such uncertainty, prior literature (e.g.,  \cite{sun2017update,sun2019wiener,TangCWYT23,bedewy2021optimal,li2024sampling,liyanaarachchi2024role,10619298,atasayar2025fresh,chen2024optimal,pan2023sampling,li2025optimal,10559951}) has predominantly focused on designing optimal sampling policies without preemption. A key takeaway is that optimal sampling under heavy-tailed delays often introduce \textit{deliberate nonzero} waiting, leading to the \textit{counterintuitive} age-optimal sampling principle: ``\textit{lazy is timely}”.  
        
        While effective, we argue that \textit{lazy sampling} alone is insufficient. The prevailing non-preemptive assumption creates a vulnerability: once a packet is in service, the system waits passively until an acknowledgment arrives. This results in severe \textit{head-of-line} blocking, preventing the transmission of fresher updates. This blockage is particularly detrimental under heavy-tailed delays (e.g., Pareto or log-normal), which can exhibit the Increasing Mean Residual Life (IMRL) property \cite{lawless2011statistical}, implying that the longer a packet has been in service, the longer it is likely to remain in service.
        
        To overcome this limitation, we introduce preemption as a control lever. Preemption allows the system to discard stale packets in service to accommodate fresher updates, introducing a nontrivial trade-off  ($i$) \textit{persisting} with the ongoing packet can utilize accumulated service effort, but risks prolonged future delays, whereas ($ii$) \textit{preempting} will sacrifice prior service progress and incur a fixed switching cost, but truncate the heavy-tailed future delays. In other words, the controller must compare the expected marginal benefit of continuing the current (already-aged) service attempt against the cost and renewal effect of restarting service with a fresh sample.

        To address this trade-off, we study a joint sampling and preemption problem under general continuous-time service-time distributions, explicitly accounting for sampling/preemption penalties. This problem is increasingly relevant for delay-sensitive applications in highly variable environments, such as space communications, drone networks, and IoT systems, where delay distributions are often heavy-tailed.
        
        \begin{figure}[t]
        \centering
        \includegraphics[width=0.9\columnwidth]{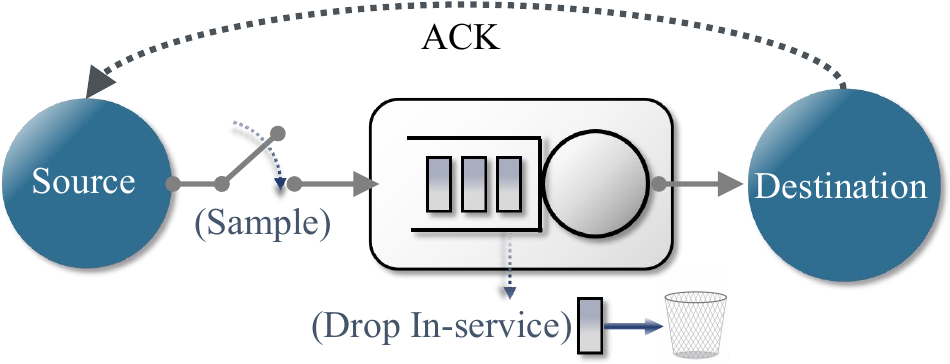}
        \caption{System Model. The source takes action $s$ (sample a fresh update) or $p$ (drop in-service update + sample a fresh update).}
        \label{fig:model}
        \vspace{-1em}
        \end{figure}
        
       \subsection{Related Works}
      A parallel line of work studies how the \emph{service discipline} shapes AoI under \emph{prescribed} queueing rules and specific stochastic models \cite{Gammaawakening,najm2017status,chen2016age,kaul2012status,costa2016age,kaul2018age,soysal2021age}. 
Here, the discipline (e.g., FCFS, LCFS with/without preemption, or priority rules) is prescribed, and the focus is to characterize the induced AoI performance. 

In a different vein, some works treat preemption as a \emph{decision variable} and seek AoI-optimal policies under simplifying assumptions (e.g., slotted decision epochs or bounded-delay models).
For instance, \cite{wang2019preemption} considers a discrete-time model with packet-level side information, while \cite{farazi2018harvesting_preemption} characterizes when preemption should be allowed under energy-harvesting constraints. \cite{asvadi2025delayed} introduces a \emph{delayed preemption} mechanism, where the optimization reduces to finding the best static deadline to trigger preemption.
\cite{banerjee_ulukus_isit24} adopts a time-slotted MDP formulation under bounded-delay modeling, with actions restricted to slot boundaries. 
In continuous time, \cite{arafa2019preemption} restricts attention to a sampling threshold and a fixed preemption cutoff. Moreover, \cite{banerjee_ulukus_isit24} and \cite{arafa2019preemption} do not jointly incorporate sampling and preemption penalties.  
In contrast, we formulate a continuous-time impulse-control problem that jointly optimizes \emph{when} to sample and \emph{when} to preempt under general service-time distributions, accounting for sampling and preemption costs. This yields an endogenous, state-dependent preemption boundary. 
Finally, while \cite{arafa2019preemption} postulates a threshold structure heuristically, we prove its optimality for exponential service times (Theorem~\ref{the1}).

        \subsection{Contributions and Novelties of This Work}
		\begin{itemize}
            \item \textbf{Impulse Control Framework for AoI}. We propose, to our knowledge, the first \textit{continuous-time impulse control framework} \cite{bensoussan1975nouvelles} for AoI optimization under general service-time distributions. Our formulation treats AoI as a controlled drifting state and sampling/preemption actions as instantaneous \textit{impulse controls}. This perspective allows us to rigorously model preemption dynamics without restricting decision epochs, and simultaneously paves the way for future \textit{hybrid control} extensions\footnote{Specifically, the drift generator can be augmented to include continuous controls such as service-rate adaptation without altering the fundamental problem structure. We leave this extension for future work and focus here on optimal sampling and preemption.}.

            \item \textbf{Integral ACOE and optimal-stopping reduction}. Instead of relying on a differential average-cost HJB--QVI and its smoothness assumptions, we derive coupled \emph{integral} average cost optimality equations (ACOE) using dynamic programming. A busy-phase invariance collapses the dynamics onto the busy-start boundary, reducing preemption to an optimal-stopping problem against the random completion time. We prove that the threshold-type policy, previously used as a heuristic (e.g., \cite{arafa2019preemption}), is in fact optimal under exponential service times.

            \item \textbf{Numerical solutions and insights:}
We develop an efficient average-cost policy-iteration algorithm tailored to heavy-tailed service times. Simulations under Pareto and log-normal service-time models show substantial gains in heavy-tailed regimes over the AoI-optimal non-preemptive sampling policy~\cite{sun2017update}, with up to $30\times$ reduction in average cost. Moreover, simulations reveal a \textit{counterintuitive} phenomenon: \textit{under preemption, delay variance can become beneficial rather than harmful.}.  
			
		\end{itemize}
        \section{System Model and Problem Formulation}
\label{sec:model}
We consider a single-source end-to-end (E2E) status update system, as depicted in Fig. \ref{fig:model}. In this setup, a source generates time-stamped status updates and transmits them to the destination over a channel with random service delay. The dynamics of the system are governed by the random delay and impulses:

\begin{itemize} \item \textbf{Random Service Delay}: The channel incurs a random service delay, modeled as a continuous non-negative random variable $Y$ with a Cumulative Distribution Function (CDF) $F_Y(\cdot)$ and density $f_Y(\cdot)$. In this paper, we assume:
\begin{equation}\label{eq1}
    \mathbb{E}[Y]<\infty,\quad\mathbb{E}[Y^2]<\infty.
\end{equation}
\item \textbf{Sampling (Idle-phase Impulse)}: When the channel is idle, the sensor decides when to \emph{generate and submit} a status update. \item \textbf{Preemption (Busy-phase Impulse)}: When the channel is busy, the sensor can intervene to \emph{discard} the in-service stale packet and transmit a fresh packet.  \end{itemize}

\subsection{Residual Life and Hazard Rate}
In status update systems allowing preemption, whether to continue the current update or preempt it depends on the conditional likelihood of completion given its \emph{service age}. We therefore characterize the residual life and hazard rate.

\begin{definition}[Residual Life and Hazard Rate \cite{shortle2018fundamentals}]
Let $Y\sim F_Y$ denote the (nonnegative) service time requirement of an update, with tail $\bar F_Y(t)\triangleq 1-F_Y(t)$. 
Fix a \textit{service age} $b\ge 0$ such that $\bar F_Y(b)>0$. The \emph{residual service time} at age $b$ is $R_b\triangleq (Y-b\,|\,Y\ge b)$, with CDF:
\begin{equation}\label{eq:residual_cdf}
\begin{aligned}
    F_{R_b}(t)&=\Pr\{R_b\le t\}
=\Pr\{Y-b\le t\mid Y\ge b\}
\\&=\frac{F_Y(b+t)-F_Y(b)}{\bar F_Y(b)},\qquad t\ge 0.
\end{aligned}
\end{equation}
The {hazard rate} is the instantaneous conditional completion rate at \textit{service age} $b$:
\begin{equation}\label{eq:hazard}
\begin{aligned}
    \lambda(b)&\triangleq 
    \lim_{\Delta\downarrow 0}\frac{\Pr\{b\le Y<b+\Delta \mid Y\ge b\}}{\Delta}\\
    &=\lim_{\Delta\downarrow 0}\frac{F_Y(b+\Delta)-F_Y(b)}{\Delta\,\bar F_Y(b)},
\end{aligned}
\end{equation}
whenever the limit exists. If $F_Y$ is absolutely continuous with density $f_Y$, then \eqref{eq:hazard} reduces to the standard form
\begin{equation}
\lambda(b)=\frac{f_Y(b)}{\bar F_Y(b)}.
\end{equation}
Moreover, under absolute continuity, $R_b$ has density $f_{R_b}(t)=f_Y(b+t)/\bar F_Y(b)$ for $t\ge 0$, and thus $\lambda(b)=f_{R_b}(0)$.
\end{definition}

The hazard rate $\lambda(b)$ quantifies the instantaneous conditional completion rate at service age $b$, and thus guides preemption decisions. 
In particular, IFR (i.e., $\lambda(b)$ is nondecreasing in $b$) means completion becomes more likely with service age, while DFR (i.e., $\lambda(b)$ is nonincreasing in $b$) means the opposite. If $Y$ is exponentially distributed, then $\lambda(b)$ is constant (memoryless), so the service age is uninformative.

\subsection{Impulse Control Formulation}
{In the following, we model the system as an impulse-controlled PDMP} \cite{davis1984piecewise}. The state follows a deterministic flow between jumps and experiences instantaneous jumps triggered by impulses or service completions.
\begin{itemize}
    \item \textbf{{System State}}: The system state is the hybrid tuple
    $
    X(t)=(\Delta(t),\,b(t),\,m(t))\in\mathcal S \triangleq \mathbb R_+^2\times\{I,B\},
    $
    where $\Delta(t)\in\mathbb R_+$ is the instantaneous AoI, $b(t)\in\mathbb R_+$ is the service age (elapsed time of the packet in service), and $m(t)\in\{I,B\}$ is the channel mode ($I$: idle, $B$: busy). Note that $b(t)=0$ whenever $m(t)=I$.
    \item \textbf{Admissible Impulses:} Two impulse actions are available, sampling $s$ and preemption $p$. The admissible action set depends on the mode:
    \begin{equation}
    \mathcal A(\Delta,b,m)=
    \begin{cases}
    \{s\}, & m=I,\\
    \{p\}, & m=B.
    \end{cases}
    \end{equation}
    \item \textbf{Deterministic Flow Between Jumps:} Between jumps (i.e., when no impulse is applied and no service completion occurs), the continuous components evolve according to $\dot \Delta(t)=1$ and $
    \dot b(t)=\mathbbm 1_{\{m(t)=B\}}$.
    \item \textbf{Stochastic Jump (Service Completion):} In the busy mode, the service completion time is governed by the hazard rate $\lambda(b)$, where $b$ is the current service age. Upon completion, a delivery occurs, and the system jumps to the idle mode, and the AoI is reset to the realized service time, which equals the service age $b$ at completion. Hence the completion jump map $\mathcal{M}_c:\mathcal{S}\rightarrow\mathcal{S}$ is
    \begin{equation}\label{eq:Mc}
        \mathcal M_c(\Delta,b,B)=(b,0,I).
    \end{equation}
    \item \textbf{Impulse Jumps (Control Actions):} The sampling impulse operator $\mathcal M_s:\mathcal S\to\mathcal S$ is
    \begin{equation}\label{eq:Ms}
        \mathcal M_s(\Delta,0,I)=(\Delta,0,B),
    \end{equation}
    meaning that a fresh update is generated and immediately enters service, resetting the service age to $0$. The preemption impulse operator $\mathcal M_p:\mathcal S\to\mathcal S$ is
    \begin{equation}\label{eq:Mp}
        \mathcal M_p(\Delta,b,B)=(\Delta,0,B),
    \end{equation}
    meaning that the ongoing service is aborted and immediately restarted with a newly sampled update; the AoI remains unchanged at the preemption time.

    \item \textbf{Infinitesimal Generator (No Impulse Applied):}
    Let $\mathcal L$ be the infinitesimal generator of the uncontrolled PDMP. For any test function $h$ that is continuously differentiable in $(\Delta,b)$ for each mode $m\in\{I,B\}$, the infinitesimal generator is defined by:
    \begin{equation}
    (\mathcal L h)(x)\triangleq \lim_{dt\downarrow 0}\frac{\mathbb E[h(X(t+dt))\mid X(t)=x]-h(x)}{dt}.
    \end{equation}
    For notational simplicity, define $h_I(\Delta)\triangleq h(\Delta,0,I)$ and $h_B(\Delta,b)\triangleq h(\Delta,b,B)$. The following lemma gives the infinitesimal generator.
    \begin{lemma}[Infinitesimal Generator]\label{lemma:infgen}
   When the channel is idle,
    \begin{equation}\label{eq:gen_idle}
        (\mathcal L h)(\Delta,0,I)=\frac{dh_I(\Delta)}{d\Delta}.
    \end{equation}
    When the channel is busy,
    \begin{equation}\label{eq:gen_busy}
    \begin{aligned}
        (\mathcal L h)(\Delta,b,B)
        &=\partial_\Delta h_B(\Delta,b)+\partial_b h_B(\Delta,b)\\
        &+\lambda(b)\big(h_I(b)-h_B(\Delta,b)\big).
    \end{aligned}
    \end{equation}
    \end{lemma}
\begin{proof}
    The proof is given in \cite[Appendix A]{aiminpreemption}.
\end{proof}    
\item \textbf{Running Cost}: The running cost function $g:\mathcal S\to\mathbb R_+$ is given by: \begin{equation}g(\Delta,b,m)=\Delta.\end{equation}
\item \textbf{Impulse Cost}: Each impulse incurs a fixed cost:
    \begin{equation}\label{eq:impulse_cost}
        K(x,a)=\kappa_a,\qquad a\in\{s,p\},
    \end{equation}
    where $\kappa_s>0$ and $\kappa_p>0$.
\end{itemize}

\subsection{Problem Formulation}
\label{subsec:objective}
Let ${\mathcal F_t}_{t\ge 0}$ be the natural filtration of $X$, i.e., $\mathcal F_t=\sigma{X(s):0\le s\le t}$. Let $\Pi$ denote the set of admissible policies.
An \emph{admissible} impulse control policy $\pi\in\Pi$ is specified by a sequence $\{(\tau_i,a_i)\}_{i\ge 1}$, where $\{\tau_i\}$ are impulse times and $\{a_i\}$ are impulse actions, such that:
(i) each $\tau_i$ is an $\{\mathcal F_t\}$-stopping time and $a_i\in \mathcal A(X(\tau_i^-))$ is $\mathcal F_{\tau_i^-}$, ensuring non-anticipativity;
(ii) $\tau_i<\tau_{i+1}$ for all $i$; and (iii) $\tau_i\uparrow\infty$ almost surely (i.e., impulses do not accumulate in finite time).

For a policy $\pi\in\Pi$, define the long-run average cost as:
\begin{equation}\label{eq:avg_cost}
\rho^\pi \triangleq
\limsup_{T\to\infty}\frac{1}{T}\,
\mathbb E^\pi\!\left[\int_0^T \Delta(t)\,dt
+\sum_{i:\,\tau_i\le T}K(X(\tau_i^-),a_i)\right],
\end{equation}
where $X(\tau_i^-)$ denotes the pre-impulse state. Our goal is to find an optimal policy $\pi^\star$ minimizing $\rho^\pi$:
\begin{problem}[Optimal Sampling and Preemption Problem]\label{p1}
\begin{equation}
    \rho=\min_{\pi\in\Pi}\rho^{\pi}.
\end{equation}
\end{problem}

		\section{Average Cost Optimality Equations}
        This section introduces optimality equations/inequalities for solving problem \ref{p1}.
        \subsection{Coupled HJB-QVI Heuristic}
        The Hamilton--Jacobi--Bellman quasi-variational inequality (HJB--QVI) provides a classical infinitesimal characterization for impulse control problems, capturing the local trade-off between continuing the process and applying an impulse. While existence and regularity results for discounted-cost QVIs are well established, a fully rigorous strong (differential) HJB--QVI theory under the \emph{long-run average cost} criterion is substantially more delicate. In particular, the relative value function is typically unbounded and may fail to be smooth at switching boundaries; accordingly, Lions and Perthame describe the strong differential formulation in the average-cost setting as ``\textit{essentially heuristic}'' \cite{bensoussan1984impulse}. The heuristic average-cost HJB--QVI takes the form \cite[Eq.~(19)]{bensoussan1984impulse}
        \begin{equation}\label{eq:heuristic_qvi_general}
\min\Big\{\, g(x)+(\mathcal L V)(x)-\rho,\ \ \mathcal M V(x)-V(x)\,\Big\}=0,\forall x,
\end{equation}
where $V$ is the relative value function normalized at a reference state (we set $V(0,0,B)=0$), $g$ is the running cost rate, $\rho$ is the optimal long-run average cost, and $\mathcal L$ is the infinitesimal generator of the uncontrolled PDMP. The intervention (impulse) operator $\mathcal M$ is defined by
\begin{equation}\label{eq:intervention_operator}
(\mathcal M V)(x)\triangleq \inf_{a\in\mathcal A(x)}\Big\{K(x,a)+V\big(\mathcal M_a(x)\big)\Big\}, \forall x,
\end{equation}
where $\mathcal A(x)$ is the admissible action set, $K(x,a)$ is the impulse cost, and $\mathcal M_a:\mathcal S\to\mathcal S$ is the impulse operator.

Although heuristic, this formulation provides insights into the structure of the optimal policy. The following Proposition instantiates this heuristic HJB-QVI for our specific problem:
\begin{proposition}[Heuristic Coupled HJB-QVI]\label{prop:heuristic_qvi}
    Let $h_I(\Delta)\triangleq V(\Delta,0,I)$ and $h_B(\Delta,b)\triangleq V(\Delta,b,B)$ denote the relative value functions in the idle and busy modes, respectively. Heuristically, they satisfy:

   \noindent\textbf{I. Idle mode (sampling vs.\ waiting):}
\begin{equation}\label{eq:qvi_idle}
\min\Big\{
\underbrace{\Delta + \partial_{\Delta}h_I(\Delta) - \rho}_{\mathcal C_I(\Delta)\,:\ \text{wait}},
\ \ 
\underbrace{\kappa_s + h_B(\Delta,0) - h_I(\Delta)}_{\mathcal S(\Delta)\,:\ \text{sample}}
\Big\}=0.
\end{equation}
Define the (heuristic) waiting and sampling regions by the complementarity conditions
\begin{equation}
\begin{aligned}
\mathcal W_I \triangleq \{\Delta:\ \mathcal C_I(\Delta)=0,\ \mathcal S(\Delta)\ge 0\},\\
\qquad
\mathcal R_I \triangleq \{\Delta:\ \mathcal S(\Delta)=0,\ \mathcal C_I(\Delta)\ge 0\}.
\end{aligned}
\end{equation}
Thus, sampling is optimal when $\Delta\in\mathcal R_I$, while waiting is optimal when $\Delta\in\mathcal W_I$.

 \noindent\textbf{II. Busy mode (preemption vs.\ continuing service):}
\begin{equation}\label{eq:qvi_busy}
\resizebox{1\hsize}{!}{$
\min\Big\{
\underbrace{\Delta + (\mathcal L h_B)(\Delta,b) - \rho}_{\mathcal C_B(\Delta,b)\,:\ \text{continue}},
\ \ 
\underbrace{\kappa_p + h_B(\Delta,0) - h_B(\Delta,b)}_{\mathcal P(\Delta,b)\,:\ \text{preempt}}
\Big\}=0,$}
\end{equation}
where from Lemma~\ref{lemma:infgen} we have
$
(\mathcal L h_B)(\Delta,b)
=\partial_\Delta h_B(\Delta,b)+\partial_b h_B(\Delta,b)
+\lambda(b)\big(h_I(b)-h_B(\Delta,b)\big).
$
Define the (heuristic) continuation and preemption regions by
\begin{equation}
\begin{aligned}
\mathcal W_B \triangleq \{(\Delta,b):\ \mathcal C_B(\Delta,b)=0,\ \mathcal P(\Delta,b)\ge 0\},\\
\qquad
\mathcal R_B \triangleq \{(\Delta,b):\ \mathcal P(\Delta,b)=0,\ \mathcal C_B(\Delta,b)\ge 0\}.
\end{aligned}
\end{equation}
Thus, preemption is optimal when $(\Delta,b)\in\mathcal R_B$, while continuing service is optimal when $(\Delta,b)\in\mathcal W_B$.
\end{proposition}

\begin{proof}
Substituting \eqref{eq:Ms}-\eqref{eq:impulse_cost} into \eqref{eq:heuristic_qvi_general} accomplishes the proof.
\end{proof}
		\subsection{Optimal Stopping-Time}\label{subsec:manifold}
       The strong (differential) HJB--QVI under the long-run average-cost criterion is often used only as an infinitesimal heuristic \cite{bensoussan1984impulse}, since the relative value function may be unbounded and non-smooth near switching boundaries, and classical smooth-fit conditions typically fail for impulse-controlled PDMPs under general service-time distributions. To avoid such regularity requirements, we work instead with an \emph{integral dynamic programming principle} \cite{stettner2022approximation}, which yields coupled integral ACOEs.   
In the busy phase, this integral formulation reduces the impulse decision to an optimal stopping problem against the random service completion time.

Moreover, our PDMP admits a structural simplification in the busy mode. 
Starting from a busy-start state $(\Delta,b,m)=(y,0,B)$ on the event ${t<Y}$ the uncontrolled flow yields $(\Delta(t),b(t))=(y+t,t)$, and hence the difference $\Delta(t)-b(t)\equiv y$ remains invariant 
until either service completion or preemption. Consequently, the busy-phase evolution can be parameterized by a single scalar $y$ 
on the \emph{busy-start boundary} $\{(\Delta,b,B): b=0\}$. We next derive coupled integral average-cost optimality equations (ACOEs) for $(\rho,h_I,v)$, 
in which the busy-phase impulse decision reduces to an optimal stopping problem with respect to the completion time.

\begin{theorem}[Integral ACOEs and Optimal-stopping Reduction]\label{thm:integral_bellman} 
The optimal average cost $\rho$ and the relative value function $(h_I(\cdot),v(\cdot))$ normalized by $v(0)=0$ satisfy the following coupled ACOE:

\noindent\textbf{I. Idle phase:} For all $\Delta\ge 0$,
\begin{equation}\label{eq:idle_integral_bellman}
\begin{aligned}
h_I(\Delta)=\rho\Delta-\frac{\Delta^2}{2}+\inf_{z\ge \Delta}\left\{\kappa_s+v(z)+\frac{z^2}{2}-\rho z\right\}.
\end{aligned}
\end{equation}

\medskip
\noindent\textbf{II. Busy phase:} For all $y\ge 0$,
\begin{equation}\label{eq:busy_integral_bellman}
\begin{aligned}
    v(y)
=\inf_{\theta\ge 0}\Big\{
Q(y,\theta;\rho)
\Big\},
\end{aligned}
\end{equation}
where
\begin{equation}\label{eq24}
\begin{aligned}
      Q(y,\theta;\rho)=  &(y-\rho)A(\theta)+J_1(\theta)
\\&+\int_{0}^\theta h_I(t)\,dF_Y(t)
+\bar F_Y(\theta)\big(\kappa_p+v(y+\theta)\big),
\end{aligned}
\end{equation}
with $A(\theta)\triangleq \mathbb{E}[\min\{Y,\theta\}]=\int_0^\theta \bar F(t)\,dt$ and $J_1(\theta)\triangleq \mathbb{E}\!\left[\int_0^{\min\{Y,\theta\}} t\,dt\right]
=\int_0^\theta t\,\bar F_Y(t)\,dt$. 
Moreover, if the infima in \eqref{eq:idle_integral_bellman}--\eqref{eq:busy_integral_bellman} are attained, then any
measurable selections of minimizers induce an optimal stationary policy as follows.

\medskip
\noindent\textbf{(a) Optimal sampling policy.}
For each $\Delta\ge 0$, if the infima are attained, any minimizers induce an optimal stationary policy:
\begin{equation}\label{eq:zstar_def}
z^\star(\Delta)\in \arg\min_{z\ge \Delta}\left\{\kappa_s+v(z)+\frac{z^2}{2}-\rho z\right\}.
\end{equation}
Starting from idle state $(\Delta,0,I)$, it is optimal to wait until the AoI reaches $z^\star(\Delta)$ (i.e., wait $u^\star(\Delta)=z^\star(\Delta)-\Delta$ time units) and then sample.

\medskip
\noindent\textbf{(b) Optimal preemption policy.}
For each busy-start AoI $y\ge 0$, choose the next preemption time based on:
\begin{equation}\label{eq:taustar_def}
\theta^\star(y)\in \arg\min_{\theta\ge 0}\Big\{Q(y,\theta;\rho)
\Big\}.
\end{equation}
Starting from a busy-start AoI $y$, it is optimal to continue service until
either completion occurs or the service age reaches $\theta^\star(y)$; if the service is
still incomplete at age $\theta^\star(y)$, then preempt and restart with a fresh update. 
\begin{proof}
   { The proof is given in \cite[Appendix B]{aiminpreemption}.}
\end{proof}
\end{theorem}
\begin{remark}
On sample paths where successive preemptions occur before any completion, the busy-start AoI
after each preemption updates as $y_{i}=y_{i-1}+\theta^\star(y_{i-1})$. If completion
occurs before $\theta^\star(\cdot)$, the system returns to the idle mode.
\end{remark}
\section{Case Study}
\subsection{Optimal Policy Structure without Preemption \texorpdfstring{\cite{sun2017update}}{}}
If preemption is not allowed (i.e., the preemption action is removed from the admissible action set), then a busy period starting from the busy-start state $(\Delta,b)=(y,0)$ must run to completion.
Accordingly, the busy-phase optimality equation \eqref{eq:busy_integral_bellman} reduces to the only feasible choice $\theta=\infty$, yielding
\begin{equation}\label{acoe:busy_np}
    v(y)=Q(y,\infty;\rho)
    =y\mathbb{E}[Y]-\rho\mathbb{E}[Y]+\frac{\mathbb{E}[Y^2]}{2}+\mathbb{E}\!\left[h_I(Y)\right],
\end{equation}
where $Y\sim F_Y$. This expression is affine in $y$ (all remaining terms are constants with respect to $y$); under the normalization $v(0)=0$, it reduces to the linear form $v(y)=y\mathbb{E}[Y]$. Substituting this into the idle-phase ACOE \eqref{eq:idle_integral_bellman} yields
\begin{equation}\label{eq28}
h_I(\Delta)=\rho\Delta-\frac{\Delta^2}{2}
+\inf_{z\ge \Delta}\left\{\kappa_s+\frac{z^2}{2}+(\mathbb{E}[Y]-\rho) z\right\}.
\end{equation}
The minimizer of the convex quadratic term in \eqref{eq28} gives the optimal sampling policy
\begin{equation}
z^{\star}(\Delta)=\max\{\Delta,\rho-\mathbb{E}[Y]\},
\end{equation}
which recovers the optimal sampling structure in \cite[Thm.~4]{sun2017update}.

\subsection{Exponential Service Times: Explicit Solutions}
When the service time is exponentially distributed, the following theorem establishes a threshold structure for the optimal sampling and preemption policies.
\begin{theorem}[Structure of Optimal Sampling and Preemption]\label{the1}
If $Y\sim\text{Exp}(\lambda)$, the following statements hold:\\
(i). The busy-phase ACOE \eqref{eq:busy_integral_bellman} admits an explicit closed-form solution: $v(y)=\frac{y}{\lambda},\forall y\ge0$;\\
(ii). There exists an optimal preemption policy that is independent of the busy-start AoI $y$; equivalently, the optimal stopping time can be chosen as a constant threshold $\theta^\star(y)\equiv \theta^\star$.\\
(iii). There exists an optimal sampling policy with a constant sampling threshold $z^{\star}(\Delta)=\max\{\Delta,\rho-\frac{1}{\lambda}\},\forall \Delta\ge0$.
\end{theorem}
\begin{proof}
    The proof is given in \cite[Appendix C]{aiminpreemption}.
\end{proof}
\begin{remark}
In \cite{arafa2019preemption}, the optimal threshold structure is postulated on heuristic grounds. Theorem~\ref{the1} rigorously establishes this structure and proves its optimality for exponential delays.
\end{remark}
		
		\section{Numerical Solutions and Simulation Results}\label{sectionVIII}
        \subsection{Policy Iteration with Heavy-Tail Acceleration}
For general service-time distributions, the optimal policy does not admit a closed-form solution; hence we compute the policy numerically using an average-cost policy iteration algorithm (see \cite[Appendix E]{aiminpreemption}). We evaluate the resulting scheme under two heavy-tailed service-time models: Pareto (Lomax) and log-normal.

\textbf{Policy evaluation equation}. Fix a stationary policy $\pi=(z_\pi,\theta_\pi)$, where $\theta_\pi(y)$ is the busy-phase threshold and
$z_\pi(\Delta)\ge \Delta$ is the post-completion waiting target.
Let $w_\pi(t)\triangleq z_\pi(t)-t\ \ge 0$. 
Under $\pi$, the idle-phase ACOE reduces to
\begin{equation}
h_I^\pi(t)=\kappa_s+v^\pi(z_\pi(t))+t\,w_\pi(t)+\tfrac12 w_\pi(t)^2-\rho^\pi w_\pi(t),
\end{equation}
obtained by substituting $z_\pi(t)=t+w_\pi(t)$ into \eqref{eq:idle_integral_bellman}.
Plugging $h_I^\pi$ into the busy-phase ACOE $v^\pi(y)=Q(y,\theta_\pi(y);\rho^\pi)$ and grouping terms yields the Poisson equation form:
\begin{equation}\label{eq:poisson}
v^\pi(y)=r_\pi(y)-\rho^\pi g_\pi(y) + (P_\pi v^\pi)(y),\qquad v^\pi(0)=0,
\end{equation}
where the definitions of $g_\pi(y)$, $r_{\pi}(y)$, and $(P_\pi v^{\pi})(y)$ are given in \cite[Appendix E]{aiminpreemption}.

\textbf{Discretization and heavy-tailed acceleration.}
Discretizing $y$ on a uniform grid $y_i=i\,dt$ turns \eqref{eq:poisson} into a sparse linear system.
At the truncation boundary, strict on-grid transitions may leave the terminal grid state without feasible positive actions. We therefore omit the boundary equation and close the system using a far-field slope constraint motivated by the linear growth of the relative value function (see \cite[Lemma 4]{aiminpreemption}), e.g., $v(y_M)-v(y_{M-1})=s\,dt$.
Heavy tails also make action truncation challenging: enforcing $\bar F_Y(\theta_{\max})\le\varepsilon$ can yield very large
$\theta_{\max}$, making uniform $\theta$-grids prohibitive. We therefore (i) use a hybrid action grid that is uniform on
$[0,\theta_{\mathrm{fine}}]$ and log-spaced beyond $\theta_{\mathrm{fine}}$, and (ii) apply a far-field linear closure
$v(y)\approx v(y_{\mathrm{cut}})+s(y-y_{\mathrm{cut}})$ for $y\ge y_{\mathrm{cut}}$, enabling stable evaluation of $v(y+\theta)$
and analytic tail handling in the idle-phase envelope minimization. We set the far-field slope $s\approx \E[Y]$ inspired by \cite[Lemma 4]{aiminpreemption}. 
        \subsection{Simulation Results}
        \begin{table}[h]
\centering
\small
\setlength{\tabcolsep}{6pt}
\renewcommand{\arraystretch}{1.15}
\caption{Average cost (lower is better). Parentheses show benchmark cost/\textsc{Tailor} cost. $\kappa_s=1$.}
\label{Table:simulation}
\begin{tabular}{lccc}
\toprule
Params. & \textsc{Tailor} (\textbf{Prop.}) & AoI-NP \cite{sun2017update} & ZW-NP\\
\midrule
$Y_P$, $\kappa_p{=}1$ &
\textbf{2.06} & 3.73 {\scriptsize\textcolor{gray}{(1.81$\times$)}} & 6.35 {\scriptsize\textcolor{gray}{(3.08$\times$)}} \\
$Y_P$, $\kappa_p{=}5$ &
\textbf{2.35} & 3.73 {\scriptsize\textcolor{gray}{(1.59$\times$)}} & 6.35 {\scriptsize\textcolor{gray}{(2.70$\times$)}} \\
$Y_{L_1}$, $\kappa_p{=}1$ &
\textbf{1.99} & 16.0 {\scriptsize{\textcolor{gray}{(8.04$\times$)}}} & 56.8 {\scriptsize\textcolor{gray}{(28.5$\times$)}} \\
$Y_{L_2}$, $\kappa_p{=}1$ &
\textbf{1.77} & 53.5 {\scriptsize\hiRatio{\textbf{(30.2$\times$)}}} & 524 {\scriptsize\textcolor{gray}{(296$\times$)}} \\
\bottomrule
\end{tabular}
\end{table}
        
		\vspace{-0.5em}
        In this section, we compare our ACOE-optimal policy (\textsc{Tailor}: \textbf{TAIL}-aware \textbf{O}ptimal p\textbf{R}eemption) with the following benchmarks. 
 \textsc{Tailor} exploits residual-life variability via an optimal-stopping preemption rule.
\begin{itemize}
    \item \textbf{Zero-wait sampling without preemption (ZW-NP):} Sample a new update immediately upon each delivery.
    \item \textbf{AoI-optimal sampling without preemption (AoI-NP) \cite{sun2017update}:} Apply the optimal no-preemption sampling rule therein.
\end{itemize}

We consider three heavy-tailed service-time distributions in Table~\ref{Table:simulation}: $Y_P\sim\mathrm{Lomax}(\sigma=1,\alpha=2.1)$ (Pareto II), $Y_{L_1}\sim\mathrm{LogNormal}(\mu=-1.31,\sigma^2=4)$, and $Y_{L_2}\sim\mathrm{LogNormal}(\mu=-2.31,\sigma^2=6)$,
where $\mathrm{Lomax}(\sigma{=}1,\,\alpha{=}2.1)$ has scale $\sigma$ and shape $\alpha$, and
$\ln Y\sim\mathcal N(\mu,\,\sigma^2)$ for the log-normal case. $Y_{L_1}$ and $Y_{L_2}$ share the same mean ($\mathbb{E}[Y]=2$), but $Y_{L_2}$ has a larger variance.

        Table~\ref{Table:simulation} demonstrates the advantage of our proposed \textsc{Tailor} under heavy-tailed service-time distributions.
In the higher-variance log-normal case $Y_{L_2}$, \textsc{Tailor} achieves an average cost of $1.77$, while AoI-NP and ZW-NP increase to $53.5$ and $524$, which are {30.2$\times$} and {296$\times$} larger, respectively.
Notably, higher variance \emph{hurts} both no-preemption benchmarks yet \emph{benefits} \textsc{Tailor} (from $1.99$ down to $1.77$), indicating a variance-driven \emph{reversal} induced by optimal preemption.
Mechanistically, preemption \emph{clips} extreme service times, yielding a lighter-tailed \emph{effective} completion-time distribution.

    \clearpage
	\bibliographystyle{IEEEtran}
	\bibliography{reference}
\clearpage
    \appendices
\normalsize
\section{Proof of Lemma \ref{lemma:infgen}}
\label{appendix:infgen}
    In the idle mode, the age increases deterministically at unit rate and $b=0$. Hence, when $X(t)=(\Delta,0,I)$, after $dt$ time units without sampling, 
    \begin{equation}
    h(X(t+dt))=h(\Delta+dt,0,I).
    \end{equation}
    Therefore, 
    \begin{equation}
    \begin{split}
    &\lim_{dt\downarrow 0}\frac{\mathbb E[h(X(t+dt))\mid X(t)=(\Delta,0,I)]-h(\Delta,0,I)}{dt} \\
    &\quad =\frac{dh_I(\Delta)}{d\Delta}.
    \end{split}
    \end{equation}
    
    In the busy mode, both $\Delta$ and $b$ increase at unit rate until completion. Over $dt$, the completion jump occurs with probability $\lambda(b)\,dt+o(dt)$, and otherwise the state drifts to $(\Delta+dt,b+dt,B)$. This yields
    \begin{equation}
    \begin{split}
    &\mathbb{E}[h(X(t+dt))\mid X(t)=(\Delta,b,B)] \\
    &\quad = (1-\lambda(b)dt+o(dt))h(\Delta+dt,b+dt,B) \\
    &\qquad + (\lambda(b)dt+o(dt))h(b,0,I).
    \end{split}
    \end{equation}
    We can rearrange this expression to finally write
    \begin{equation}
    \begin{split}
    &\lim_{dt\downarrow 0}\frac{\mathbb E[h(X(t+dt))\mid X(t)=(\Delta,b,B)]-h(\Delta,b,B)}{dt} \\
    &=\partial_\Delta h_B(\Delta,b)+\partial_b h_B(\Delta,b) +\lambda(b)\big(h_I(b)-h_B(\Delta,b)\big).
    \end{split}
    \end{equation}

\section{{Proof of Theorem~\ref{thm:integral_bellman}}}\label{app:integral_bellman_proof}
\subsection{Useful Lemmas}
\begin{lemma}[Idle-phase reduction to deterministic thresholds]\label{lem:idle_os_equiv}
Fix an idle-start state $(\Delta,0,I)$. In idle mode, the only admissible impulse is sampling.
Let $\tau\in[0,\infty]$ be the (possibly history-dependent) first sampling time after the idle start, with $\tau=\infty$
meaning ``never sample''. Then the one-step objective (relative running cost up to $\tau$ plus continuation bias) is
\begin{equation}\label{eq:idle_os_decomp}
\begin{aligned}
&Q_I(\Delta,\tau;\rho)\triangleq\\
&\mathbb{E}\!\left[\int_{0}^{\tau}(\Delta+t-\rho)\,dt+\kappa_s+v(\Delta+\tau)\mathbf 1_{\{\tau<\infty\}}\right].
\end{aligned}
\end{equation}

Moreover:
\begin{enumerate}
\item[\textup{(i)}] Any admissible \emph{deterministic} within-idle rule is equivalent to a deterministic waiting time.
\item[\textup{(ii)}] For any (possibly randomized) $\tau$, $Q_I(\Delta,\tau;\rho)\ge \inf_{w\ge0} Q_I(\Delta,w;\rho)$; hence
randomization cannot improve upon deterministic waiting.
\item[\textup{(iii)}] For deterministic $w\ge0$,
\begin{equation}
Q_I(\Delta,w;\rho)=\rho\Delta-\frac{\Delta^2}{2}+\kappa_s+v(\Delta+w)+\frac{z^2}{2}-\rho z,
\end{equation}
so $\inf_{w\ge0}Q_I(\Delta,w;\rho)$ is equivalent to the infimum over $z\ge\Delta$ in \eqref{eq:idle_integral_bellman}.
\end{enumerate}
\end{lemma}

\begin{proof}
On $[0,\tau)$ there is no jump and $\Delta(t)=\Delta+t$ deterministically, so \eqref{eq:idle_os_decomp} follows.

(i) Within an idle period there is no exogenous randomness; the within-idle observation history up to time $t$ is the same
deterministic trajectory for all sample paths. Hence any deterministic admissible rule must prescribe the same decision at
each $t$, and therefore selects a deterministic sampling time $w\in[0,\infty]$; setting $z=\Delta+w$ gives the threshold
form.

(ii) Since the within-idle evolution is deterministic, $Q_I(\Delta,\tau;\rho)=\mathbb E[g(\tau)]$ for a deterministic
function $g(\cdot)$, hence it is a convex combination of deterministic values and is at least $\inf_{w\ge0}g(w)$.

(iii) For deterministic $w$,
\begin{equation}
\int_0^w(\Delta+t-\rho)\,dt=(\Delta-\rho)w+\frac{w^2}{2},
\end{equation}
and substituting $z=\Delta+w$ yields the expression.
\end{proof}


\begin{lemma}[Busy-phase reduction to deterministic thresholds]\label{lem:busy_os_equiv}
Fix a busy-start state $(y,0,B)$. In busy mode, the only admissible impulse is preemption.
Let $Y\ge 0$ denote the completion time of the current service attempt. Let $\tau$ be the first preemption time within the \emph{current} attempt, with the convention
$\tau=\infty$ if no preemption occurs before completion. Assume $\tau$ is a stopping time with respect to the natural
within-attempt filtration generated by the completion observation $\{\mathbf 1_{\{Y\le t\}}:t\ge0\}$. Define $T\triangleq \min\{Y,\tau\}$.
Then the one-step objective (relative running cost up to $T$ plus continuation bias) is
\begin{equation}\label{eq:busy_os_decomp}
\begin{aligned}
Q(y,\tau;\rho)\triangleq
\mathbb{E}\!\bigg[\int_{0}^{T}(\Delta(t)-\rho)\,dt
+ h_I(Y)\mathbf{1}_{\{Y\le\tau\}}\\
+ (\kappa_p+v(y+\tau))\mathbf{1}_{\{Y>\tau\}}\bigg],
\end{aligned}
\end{equation}
where on $\{t<Y\}$ the state evolves deterministically as $\Delta(t)=y+t$, $b(t)=t$, $m(t)=B$. Moreover, for any admissible \emph{deterministic} within-attempt rule,
there exists a deterministic threshold $\theta\in[0,\infty]$ such that within the current attempt
$
\min\{Y,\tau\}=\min\{Y,\theta\},\qquad
\mathbf 1_{\{Y\le\tau\}}=\mathbf 1_{\{Y\le\theta\}},\qquad
\mathbf 1_{\{Y>\tau\}}=\mathbf 1_{\{Y>\theta\}},
$
and hence $Q(y,\tau;\rho)=Q(y,\theta;\rho)$. Finally, for any deterministic $\theta\in[0,\infty]$, evaluating \eqref{eq:busy_os_decomp} yields \eqref{eq24}.
\end{lemma}

\begin{proof}
Let $T=\min\{Y,\tau\}$. Up to $T$ there is no jump, hence on $[0,T)$,
$(\Delta(t),b(t),m(t))=(y+t,t,B)$. At $T$, either completion occurs first ($Y\le\tau$) and the system jumps to idle with
bias $h_I(Y)$, or preemption occurs first ($Y>\tau$), incurring cost $\kappa_p$ and restarting busy with AoI $y+\tau$ and
bias $v(y+\tau)$. This gives \eqref{eq:busy_os_decomp}.

Consider any deterministic admissible rule within the current attempt.
Fix $t\ge0$ and restrict to the survival event $\{Y>t\}$. On $\{Y>t\}$, the within-attempt observation history up to time $t$
contains only the deterministic trajectory $(y+s,s,B)_{0\le s\le t}$ and the information ``no completion yet'';
consequently, the rule must prescribe a single action at age $t$ conditional on survival, denoted
$a(t)\in\{\textsf{continue},\textsf{preempt}\}$.
Define the deterministic threshold
\begin{equation}
\theta \triangleq \inf\{t\ge0:\ a(t)=\textsf{preempt}\}\in[0,\infty].
\end{equation}
If $Y\le\theta$, then $a(t)=\textsf{continue}$ for all $t<\theta$, so no preemption occurs before completion; thus
$Y\le\tau$ and the completion indicators match. If $Y>\theta$, then the rule preempts at age $\theta$, hence $\tau=\theta$ on
$\{Y>\theta\}$, so the preemption indicators match. Therefore $\min\{Y,\tau\}=\min\{Y,\theta\}$ and $Q(y,\tau;\rho)=Q(y,\theta;\rho)$ within the current attempt.

Fix deterministic $\theta\in[0,\infty)$ and let $T=\min\{Y,\theta\}$. Using
$T=\int_0^\theta \mathbf 1_{\{Y>t\}}\,dt$, we have
\begin{equation}\label{eq41}
\begin{aligned}
\mathbb E[T]&=\int_0^\theta \bar F_Y(t)\,dt=A(\theta),\\
\mathbb E\!\left[\int_0^T t\,dt\right]&=\int_0^\theta t\,\bar F_Y(t)\,dt=J_1(\theta).
\end{aligned}
\end{equation}
Also, by the Lebesgue--Stieltjes integral for $F_Y$ and completion-precedence,
\begin{equation}\label{eq42}
\begin{aligned}
\mathbb E[h_I(Y)\mathbf 1_{\{Y\le\theta\}}]&=\int_{[0,\theta]}h_I(t)\,dF_Y(t),
\\
\mathbb P(Y>\theta)&=\bar F_Y(\theta).
\end{aligned}
\end{equation}
Substituting \eqref{eq41} and \eqref{eq42} into \eqref{eq:busy_os_decomp} yields \eqref{eq24}.
\end{proof}



\subsection{Formal Proof}
For any policy $\pi$ and any decision epoch $S_N$,
\begin{equation}
\begin{aligned}
&\int_0^{S_N}\Delta(t)\,dt+\sum_{\tau_i\le S_N}K(\cdot)
=
\\&\rho\,S_N+\left(\int_0^{S_N}(\Delta(t)-\rho)\,dt+\sum_{\tau_i\le S_N}K(\cdot)\right),
\end{aligned}
\end{equation}
so controlling the centered growth rate of $\int_0^{S_N}(\Delta-\rho)\,dt+\sum K$ is equivalent to controlling the average
cost rate $\rho^\pi$.

Fix any $\pi\in\Pi$ and let $0=S_0<S_1<S_2<\cdots$ be the successive event times under $\pi$, with $X_N=X(S_N)$.
At each epoch $S_N$, the system is either in idle $(\Delta_N,0,I)$ or in busy $(y_N,0,B)$.
In both cases, using the ACOE plus Lemmas~\ref{lem:idle_os_equiv}--\ref{lem:busy_os_equiv} yields the Bellman inequality
\begin{equation}\label{eq:bellman_ineq_epochs}
\begin{aligned}
&h(X_N)\le \\&\mathbb E^\pi\!\left[\int_{S_N}^{S_{N+1}}(\Delta(t)-\rho)\,dt+K_N+h(X_{N+1})\ \Big|\ \mathcal F_{S_N}\right],
\end{aligned}
\end{equation}
where $K_N\in\{0,\kappa_s,\kappa_p\}$ is the impulse cost incurred at time $S_{N+1}$. Taking expectations in \eqref{eq:bellman_ineq_epochs} and summing $N=0,\dots,n-1$ yields
\[
\mathbb E^\pi\!\left[\int_{0}^{S_n}(\Delta(t)-\rho)\,dt+\sum_{N=0}^{n-1}K_N\right]
\ge h(X_0)-\mathbb E^\pi[h(X_n)].
\]
Equivalently,
\begin{equation}
\begin{aligned}
\mathbb E^\pi\!\left[\int_{0}^{S_n}\Delta(t)\,dt+\sum_{\tau_i\le S_n}K(\cdot)\right]
\ge \\\rho\,\mathbb E^\pi[S_n]+h(X_0)-\mathbb E^\pi[h(X_n)].
\end{aligned}
\end{equation}
Dividing by $\mathbb E^\pi[S_n]$ and letting $n\to\infty$, we have
\begin{equation}
    \rho^\pi\ge \rho,\qquad \pi\in\Pi.
\end{equation}

Let $\pi^\star$ be the stationary policy induced by measurable minimizers $z^\star(\cdot)$ and $\theta^\star(\cdot)$ in
\eqref{eq:zstar_def}--\eqref{eq:taustar_def}. Under $\pi^\star$, the chosen actions attain the infima in
\eqref{eq:idle_integral_bellman} and \eqref{eq:busy_integral_bellman}, so \eqref{eq:bellman_ineq_epochs} holds with equality
at every epoch, implying $\rho^{\pi^\star}=\rho$. 

\section{Proof of Theorem \ref{the1}}\label{proof:the1}
\subsection{Useful Lemmas}
\begin{lemma}[Exponential service: translation-invariant stopping equation]\label{lemma3}
Let $Y\sim\mathrm{Exp}(\lambda)$ with $\lambda>0$ and define $u(y)\triangleq v(y)-y/\lambda$.
Then the busy-phase ACOE is equivalent to the fixed-point equation
\begin{equation}\label{eq:u_fp}
u(y)=\inf_{\tau\ge0}\Big\{G(\tau)+e^{-\lambda\tau}u(y+\tau)\Big\},\qquad y\ge0,
\end{equation}
where
\begin{equation}
G(\tau)=\int_{0}^{\tau} e^{-\lambda t}\big(t-\rho+\lambda h_I(t)\big)\,dt
+e^{-\lambda\tau}\Big(\kappa_p+\frac{\tau}{\lambda}\Big).
\end{equation}
Moreover, for any $c\ge0$, the shifted function $u_c(y)\triangleq u(y+c)$ satisfies the same equation.
\end{lemma}
\begin{proofsketch}
    Starting from the busy-phase ACOE \eqref{eq:busy_integral_bellman}, set $u(y)\triangleq v(y)-y/\lambda$ and
substitute $v(y)=u(y)+y/\lambda$ and $v(y+\tau)=u(y+\tau)+(y+\tau)/\lambda$. After regrouping terms, the only
$y$-dependent remainder is $\big(e^{-\lambda\tau}-1\big)\frac{y}{\lambda}$, which cancels with the linear shift
in the definition of $u$ by using $1-e^{-\lambda\tau}=\int_0^\tau \lambda e^{-\lambda t}\,dt$. This yields \eqref{eq:u_fp} with the stated $G(\tau)$.
Finally, since the right-hand side of \eqref{eq:u_fp} depends on $y$ only through $u(y+\tau)$, any shift
$u_c(y)=u(y+c)$ satisfies the same equation.
\end{proofsketch}

\begin{lemma}[Linear-growth upper bound of $v$]\label{lem:linear_growth}For any distributions $Y\sim F_Y$,
there exists a constant $C_0\in\mathbb{R}$ independent of $y$ such that
\begin{equation}\label{eq:linear_growth_bound}
v(y)\le y\,\mathbb{E}[Y]+C_0,\qquad \forall y\ge 0.
\end{equation}
\end{lemma}
\begin{proofsketch}
By optimality of \eqref{eq:busy_integral_bellman} and never preemption case given in \eqref{acoe:busy_np}, we establish the following upper bound:
\begin{equation}
v(y)\le y\mathbb{E}[Y]+\frac{\mathbb{E}[Y^2]}{2}-\rho\mathbb{E}[Y]+\mathbb{E}[h_I(Y)],
\end{equation}
which yields \eqref{eq:linear_growth_bound} with
$C_0=\frac{\mathbb{E}[Y^2]}{2}-\rho\mathbb{E}[Y]+\mathbb{E}[h_I(Y)]$.
\end{proofsketch}
\begin{lemma}[Monotonicity of the $v(y)$]\label{lem:v_monotone}
$v(\cdot)$ is nondecreasing on $\mathbb{R}_+$; in particular, $v(y)\ge 0$ for all $y\ge0$.
\end{lemma}
\begin{proof}
See Appendix \ref{appD}

\end{proof}

\subsection{Formal Proof}

\smallskip
\noindent\emph{Step 1: shift-invariance of the fixed-point equation.}
For any $c\ge0$, define $u_c(y)\triangleq u(y+c)$. Since $G(\tau)$ does not depend on $y$, $u_c$ satisfies the same fixed-point equation \eqref{eq:u_fp}:
\begin{equation}
u_c(y)=\inf_{\tau\ge0}\Big\{G(\tau)+e^{-\lambda \tau}u_c(y+\tau)\Big\}.
\end{equation}

\smallskip
\noindent\emph{Step 2: $\varepsilon$-optimal dynamic-programming comparison (``sandwich'') for $d_c(y)$.}
Fix $c\ge0$ and define $u_c(y)\triangleq u(y+c)$ and $d_c(y)\triangleq u_c(y)-u(y)$.
Fix $y\ge0$ and a total error budget $\varepsilon>0$. We construct a summable error sequence, e.g.
\begin{equation}
\varepsilon_k \triangleq 2^{-(k+1)}\varepsilon,\qquad k=0,1,2,\dots,
\end{equation}
{so that} $\sum_{k=0}^\infty \varepsilon_k=\varepsilon$. We select the actions along an \emph{admissible} $\{\varepsilon_k\}$-optimal policy:
at each visited state $y_k$ (resp.\ $\hat y_k$), choose $\tau_k$ (resp.\ $\hat\tau_k$) to be $\varepsilon_k$-optimal.
This is always possible by the definition of the infimum in \eqref{eq:u_fp}.

\smallskip
\noindent\emph{Upper sandwich.}
Set $y_0=y$. For each $k\ge0$, since $u$ satisfies \eqref{eq:u_fp}, there exists an $\varepsilon_k$-optimal action
$\tau_k\ge0$ (chosen at state $y_k$) such that
\begin{equation}\label{eq:epsk_opt_u}
u(y_k)\;\ge\;G(\tau_k)+e^{-\lambda \tau_k}u(y_k+\tau_k)-\varepsilon_k.
\end{equation}
Using the same $\tau_k$ as a feasible action for $u_c(y_k)$ gives
\begin{equation}\label{eq:feasible_uc_k}
u_c(y_k)\;\le\;G(\tau_k)+e^{-\lambda \tau_k}u_c(y_k+\tau_k).
\end{equation}
Subtracting \eqref{eq:epsk_opt_u} from \eqref{eq:feasible_uc_k} yields
\begin{equation}\label{eq:dc_upper_epsk}
d_c(y_k)\;\le\;e^{-\lambda \tau_k}\,d_c(y_{k+1})+\varepsilon_k,
\qquad y_{k+1}\triangleq y_k+\tau_k.
\end{equation}

\smallskip
\noindent\emph{Lower sandwich.}
Similarly, define $\hat y_0=y$ and for each $k\ge0$ choose an $\varepsilon_k$-optimal action $\hat\tau_k\ge0$
for $u_c(\hat y_k)$ such that
\begin{equation}\label{eq:epsk_opt_uc}
u_c(\hat y_k)\;\ge\;G(\hat\tau_k)+e^{-\lambda \hat\tau_k}u_c(\hat y_k+\hat\tau_k)-\varepsilon_k.
\end{equation}
Using the same $\hat\tau_k$ as a feasible action for $u(\hat y_k)$ gives
\begin{equation}\label{eq:feasible_u_k}
u(\hat y_k)\;\le\;G(\hat\tau_k)+e^{-\lambda \hat\tau_k}u(\hat y_k+\hat\tau_k).
\end{equation}
Subtracting \eqref{eq:feasible_u_k} from \eqref{eq:epsk_opt_uc} yields
\begin{equation}\label{eq:dc_lower_epsk}
d_c(\hat y_k)\;\ge\;e^{-\lambda \hat\tau_k}\,d_c(\hat y_{k+1})-\varepsilon_k,
\qquad \hat y_{k+1}\triangleq \hat y_k+\hat\tau_k.
\end{equation}

\smallskip
\noindent\emph{Step 3: iteration and vanishing of tail terms, then let $\varepsilon\downarrow 0$.}
Iterating \eqref{eq:dc_upper_epsk} for $k=0,1,\dots,n-1$ gives
\begin{equation}\label{eq:iterate_upper_epsk}
d_c(y)\;\le\;e^{-\lambda S_n}\,d_c(y_n)\;+\;\sum_{k=0}^{n-1}e^{-\lambda S_k}\varepsilon_k,
\end{equation}
where
\begin{equation}
    S_n\triangleq\sum_{j=0}^{n-1}\tau_j,\ \ S_0=0.
\end{equation}
Since $e^{-\lambda S_k}\le 1$ and $\sum_{k\ge0}\varepsilon_k=\varepsilon$, we have
\begin{equation}\label{eq:error_bound}
\sum_{k=0}^{n-1}e^{-\lambda S_k}\varepsilon_k\;\le\;\sum_{k=0}^{n-1}\varepsilon_k\;\le\;\varepsilon.
\end{equation}
By construction, the sequence $\{\tau_k\}_{k\ge0}$ is generated by an admissible policy. The same holds for $\left\{\hat{\tau}_k\right\}$ along an admissible policy for $u_c$.
Hence the associated impulse times cannot accumulate in finite time, which implies
\begin{equation}
S_n=\sum_{j=0}^{n-1}\tau_j \uparrow \infty \quad \text{as } n\to\infty,
\end{equation}
and therefore $e^{-\lambda S_n}\to0$.

Next, we control $d_c(y_n)$. By Lemma~\ref{lem:linear_growth} and $\mathbb{E}[Y]=1/\lambda$ (for $Y\sim{\rm Exp}(\lambda)$),
there exists $C_0$ such that $v(y)\le y/\lambda+C_0$, hence $u(y)=v(y)-y/\lambda\le C_0$.
By Lemma~\ref{lem:v_monotone} and $v(0)=0$, we have $v(y)\ge0$ and thus $u(y)\ge -y/\lambda$.
Therefore, for each fixed $c\ge0$,
\begin{equation}\label{eq:dc_linear_bound_final}
\begin{aligned}
|d_c(y)|=\big|u(y+c)-u(y)\big|
&\le \max\Big\{C_0+\frac{y}{\lambda},\ C_0+\frac{y+c}{\lambda}\Big\}
\\&\le C_1(c)+\frac{y}{\lambda},
\end{aligned}
\end{equation}
for some constant $C_1(c)$ independent of $y$. Since $y_n=y+S_n$, we have
\[
|d_c(y_n)|\le C_1(c)+\frac{y+S_n}{\lambda}.
\]
Consequently,
\begin{equation}
\begin{aligned}
e^{-\lambda S_n}|d_c(y_n)|
&\le C_1(c)e^{-\lambda S_n}+\frac{y}{\lambda}e^{-\lambda S_n}+\frac{1}{\lambda}S_ne^{-\lambda S_n}
\\&\;\xrightarrow[n\to\infty]{}\;0,
\end{aligned}
\end{equation}
since $\lim_{s\to\infty} s e^{-\lambda s}=0$.

Letting $n\to\infty$ in \eqref{eq:iterate_upper_epsk} and using \eqref{eq:error_bound} yields
\[
d_c(y)\le \varepsilon.
\]
Since $\varepsilon>0$ is arbitrary, letting $\varepsilon\downarrow0$ gives $d_c(y)\le0$.

Repeating the same iteration argument using \eqref{eq:dc_lower_epsk} yields $d_c(y)\ge -\varepsilon$,
and then letting $\varepsilon\downarrow0$ gives $d_c(y)\ge0$.
Therefore $d_c(y)=0$ for all $y\ge0$, i.e., $u(y+c)=u(y)$ for all $y,c\ge0$.
Hence $u$ is constant on $[0,\infty)$. With the normalization $v(0)=0$, we have $u(0)=0$,
so $u\equiv0$ and thus
\begin{equation}
v(y)=\frac{y}{\lambda},\qquad y\ge0.
\end{equation}

\smallskip
\noindent\emph{Step 4: independence of the optimal preemption threshold from $y$.}
Substituting $v(y)=y/\lambda$ into the exponential busy objective yields
\begin{equation}
Q(y,\tau;\rho)=\frac{y}{\lambda}+G(\tau),
\end{equation}
where $G(\tau)$ does not depend on $y$. Hence $\arg\min_{\tau\ge0}Q(y,\tau;\rho)=\arg\min_{\tau\ge0}G(\tau)$ is independent of $y$, and an optimal preemption rule can be taken as a constant threshold $\tau^\star$.

\section{Proof of Lemma \ref{lem:v_monotone}} \label{appD}
Fix $c\ge0$ and set $y_2=y_1+c$. We compare the relative performance starting from the busy-start states
$(\Delta,b,m)=(y_1,0,B)$ and $(y_2,0,B)$.

\smallskip
\noindent\textbf{Step 1 (Shifted policy).}
Let $\pi$ be an arbitrary admissible policy for the system started from $y_2$.
Construct a shifted policy $\pi^{(-c)}$ for the system started from $y_1$ as follows: in the busy mode, when the current state is $(\Delta,b,B)$, let $\pi^{(-c)}$ apply the same impulse decision that $\pi$ would apply at the shifted state $(\Delta+c,b,B)$.
This is well-defined since admissibility depends only on the mode (idle vs.\ busy), so both $(\Delta,b,B)$ and $(\Delta+c,b,B)$ admit the same action (preemption).
After a service completion, both systems enter the same idle state under the coupling below; hence no shift is needed thereafter.
By construction, $\pi^{(-c)}$ is non-anticipative and admissible.

\smallskip
\noindent\textbf{Step 2 (Pathwise coupling up to the first delivery).}
Couple the two systems on a common probability space by using the same i.i.d.\ service-time samples for each (re)start and, if policies are randomized, the same auxiliary randomization.
Let $T_{\rm del}$ denote the time of the first service completion (delivery) after the busy start.
By construction, whenever the service age equals $b$, the policy $\pi$ in the $y_2$-system and the policy $\pi^{(-c)}$ in the $y_1$-system make identical impulse decisions; hence the impulse times also coincide up to $T_{\rm del}$ and the deterministic between-impulse evolution implies $b_{y_1}(t)=b_{y_2}(t)$ for $t\le T_{\rm del}$.
Consequently,
\begin{equation}
\Delta_{y_2}(t)=\Delta_{y_1}(t)+c,\qquad \text{for all } t< T_{\rm del}.
\end{equation}
At time $T_{\rm del}$, both systems jump to the idle mode with AoI reset to the realized service time, which is identical under the coupling; hence their post-delivery states coincide, and their future trajectories (and costs) coincide when driven by the same randomness.

\smallskip
\noindent\textbf{Step 3 (Cost comparison).}
Since the impulse times and types coincide under the coupling, the impulse costs are identical. With running cost $g(\Delta,b,m)=\Delta$, the only difference in the accumulated running cost up to $T_{\rm del}$ is the constant AoI offset $c$, so
\begin{equation}
\int_0^{T_{\rm del}} \Delta_{y_2}(t)\,dt
=\int_0^{T_{\rm del}} \Delta_{y_1}(t)\,dt + c\,T_{\rm del}.
\end{equation}
The same identity holds for the relative running cost with baseline subtraction $\rho$:
\begin{equation}
\int_0^{T_{\rm del}} \big(\Delta_{y_2}(t)-\rho\big)\,dt
=\int_0^{T_{\rm del}} \big(\Delta_{y_1}(t)-\rho\big)\,dt + c\,T_{\rm del}.
\end{equation}
Moreover, the continuation bias at $T_{\rm del}$ is the same for both systems because the post-delivery state is identical.

Let $J(\pi;y)$ denote the one-step relative cost functional in the integral DPP from the busy-start state $(y,0,B)$ up to the first delivery under policy $\pi$.
By admissibility, impulses do not accumulate in finite time, so $T_{\rm del}<\infty$; under $\mathbb{E}[Y]<\infty$ we have $\mathbb{E}[T_{\rm del}]<\infty$.
Taking expectations in the above pathwise comparison yields
\begin{equation}
J(\pi;y_2)=J(\pi^{(-c)};y_1)+c\,\mathbb{E}[T_{\rm del}]
\;\ge\; J(\pi^{(-c)};y_1).
\end{equation}

\smallskip
\noindent\textbf{Step 4 (Optimize).}
Since the above holds for any admissible $\pi$ at $y_2$, taking the infimum over $\pi$ gives
\begin{equation}
v(y_2)\ge v(y_1),
\end{equation}
and thus $v(\cdot)$ is nondecreasing. With the normalization $v(0)=0$, we further have $v(y)\ge0$ for all $y\ge0$.

\section{Policy Iteration with Assertion Heavy Tail}\label{appE}
The $g_\pi(y)$, $g_\pi(y)$, and $g_\pi(y)$ in the Poisson equation \eqref{eq:poisson} are given as:
\begin{subequations} \label{eq:bellman_components}
\begin{align}
g_\pi(y) &= A(\tau_\pi(y)) + \int_0^{\tau_\pi(y)} f_Y(t) w_\pi(t) dt, \\
r_\pi(y) &= \begin{multlined}[t] 
            y A(\tau_\pi(y)) + J_1(\tau_\pi(y)) + \bar{F}_Y(\tau_\pi(y)) \kappa_p \\ 
            + \int_0^{\tau_\pi(y)} \! f_Y(t) \big( \kappa_s + t w_\pi(t) + \tfrac{1}{2} w_\pi(t)^2 \big) dt, 
            \end{multlined} \\
(P_\pi v^\pi)(y) &= \begin{multlined}[t] 
            \int_0^{\tau_\pi(y)} f_Y(t) v^\pi(z_\pi(t)) dt \\ 
            + \bar{F}_Y(\tau_\pi(y)) v^\pi(y + \tau_\pi(y)).
            \end{multlined}
\end{align}
\end{subequations}
The details of the Algorithm are shown in \ref{alg:pi-ultra-compact}.

\begin{algorithm*}[t]
\caption{Policy Iteration with Heavy-Tail Acceleration}
\label{alg:pi-ultra-compact}
\KwIn{State grid $y_i=i\,dt$ on $[0,y_{\mathrm{cut}}]$; action grid $t_j=j\,dt$ on $[0,\tau_{\max}]$;
$(f_Y,\bar F_Y)$; $(\kappa_s,\kappa_p)$; slope $s$; tolerances $\varepsilon_v,\varepsilon_\rho$.}
\KwOut{$\rho$, $v(\cdot)$ with $v(0)=0$, idle map $z(\cdot)$, busy map $\theta(\cdot)$.}

Precompute $A(t_j)=\int_0^{t_j}\bar F_Y(u)\,du$ and $J_1(t_j)=\int_0^{t_j}u\bar F_Y(u)\,du$;\;
Build hybrid candidates $\mathcal{J}$ (uniform on $[0,\theta_{\mathrm{fine}}]$, log-spaced beyond) and quantize to grid\;
Initialize $v^{(0)}\equiv 0$, $\rho^{(0)}\leftarrow 0$, and $\theta^{(0)}(\cdot)$\;

\For{$k=0,1,2,\dots$}{
\textbf{Idle map (suffix-min envelope)}\;
Compute $\phi_i=\kappa_s+v^{(k)}(y_i)+\tfrac12 y_i^2-\rho^{(k)}y_i$\;
Backward scan to get $m_i=\min_{r\ge i}\phi_r$ and $a_i\in\arg\min_{r\ge i}\phi_r$\;
Set $z^{(k)}(y_i)\leftarrow y_{a_i}$ and compute $h_I^{(k)}(t_j)=\rho^{(k)}t_j-\tfrac12 t_j^2+m_{i_0}$ for $t_j\le y_{\mathrm{cut}}$
(with $i_0=\mathrm{round}(t_j/dt)$), otherwise use far-field closure\;
Precompute $I_{fh}^{(k)}(t_j)=\int_0^{t_j} f_Y(t)h_I^{(k)}(t)\,dt$ by cumulative quadrature\;

\textbf{Busy improvement (hybrid $\theta$, strict on-grid)}\;
\For{$i=0,\dots,M$}{
Set $\theta^{(k+1)}(y_i)\in\arg\min_{j\in\mathcal{J}} Q^{(k)}(y_i,t_j)$, where $Q^{(k)}$ is \eqref{eq24}
evaluated using $I_{fh}^{(k)}$ and $\hat v^{(k)}(y)$ (round-to-grid lookup if $y\le y_{\mathrm{cut}}$, else far-field
$v^{(k)}(y_M)+s(y-y_M)$)\;
}

\textbf{Policy evaluation (Poisson equation)}\;
Form the sparse system $v=r_\pi-\rho g_\pi+P_\pi v$ induced by $(z^{(k)},\theta^{(k+1)})$,
with $(g_\pi,r_\pi,P_\pi)$ given by \eqref{eq:bellman_components}\;
Impose $v(0)=0$ and close the right boundary by $v(y_M)-v(y_{M-1})=s\,dt$; solve for $(\rho^{(k+1)},v^{(k+1)})$\;

\If{$\|v^{(k+1)}-v^{(k)}\|_\infty\le \varepsilon_v$ \textbf{and} $|\rho^{(k+1)}-\rho^{(k)}|\le \varepsilon_\rho$
\textbf{and} $\theta^{(k+1)}(\cdot)=\theta^{(k)}(\cdot)$ on-grid}{
\textbf{break}\;
}
}
\end{algorithm*}




		%
						\vspace{-1cm}
						\IEEEbiographynophoto{Aimin Li}
						(Member, IEEE) received the B.S. degree (Best
Thesis Award) and the Ph.D. degree (Awarded Best Dissertation Nomination)
in electronic engineering from Harbin Institute of Technology (Shenzhen)
in 2020 and 2025, respectively. From 2023 to 2024, he was a visiting
researcher with the Institute for Infocomm Research (I2R), Agency for Science, Technology, and Research (A*STAR), Singapore. He is currently a postdoctoral researcher at the Middle East Technical University (METU), Ankara, Turkey. His research interests include advanced channel coding, information theory, and wireless communications. He has served as a reviewer for IEEE TIT, IEEE JSAC, IEEE ISIT, among others, and as a session
chair for IEEE Information Theory Workshop 2024 and IEEE Globecom 2024.

		\IEEEbiographynophoto{Yiğit İnce}
        received the B.S. degree in Electrical and Electronics Engineering from Middle East Technical University (METU) in 2025, graduating as the valedictorian. He is currently pursuing the M.Sc. degree in the same department. His research interests include random access protocols, wireless communications, non-terrestrial networks, and decision-making theory.
        
		\IEEEbiographynophoto{Elif Uysal} (Fellow, IEEE) is a Professor of Electrical and Electronics Engineering at the Middle East Technical University (METU), in Ankara, Turkiye. She received the Ph.D. degree from Stanford University in 2003, the S.M. degree from the Massachusetts Institute of Technology (MIT) in 1999 and the B.S. in 1997 from METU. She was elected IEEE Fellow for ``pioneering contributions to energy-efficient and low-latency communications'', and Fellow of the Artificial Intelligence Industry Alliance in 2022. She received an ERC Advanced Grant in 2024 for her project GO SPACE (Goal Oriented Networking for Space). Dr. Uysal is also a recipient TUBITAK BIDEB National Pioneer Researcher Grant (2020), Science Academy of Turkey Young Scientist Award (2014). She has chaired the METU Parlar Foundation for Education and Research since 2022. She founded FRESHDATA Technologies in 2022.

\end{document}